\documentclass[12pt,american,pra,reprint,superscriptaddress,nofootinbib]{revtex4-2}
\usepackage{amsmath}
\usepackage{dsfont,mathrsfs,bm,xcolor,amsthm}
\usepackage{amstext}
\usepackage{breakurl}
\usepackage{braket}
\usepackage[unicode=true,pdfusetitle,
bookmarks=true,bookmarksnumbered=false,bookmarksopen=false,
breaklinks=false,pdfborder={0 0 0},pdfborderstyle={},backref=false,colorlinks=true]
{hyperref}

\theoremstyle{plain}
\newtheorem{thm}{Theorem}[section]
\newtheorem{prop}[thm]{Proposition}
\newtheorem{coro}[thm]{Corollary}
\newtheorem{lem}[thm]{Lemma}

\theoremstyle{definition}
\newtheorem{defn}[thm]{Definition}
\newtheorem{exmp}[thm]{Example}

\begin{document}

\title[The operational foundations of PT-symmetry]{The operational foundations of PT-symmetric and quasi-Hermitian quantum theory}

\author{Abhijeet Alase}
\email{abhijeet.alase1@ucalgary.ca}
\affiliation{Department of Physics \& Astronomy, University of Calgary, Calgary, AB T2N 1N4, Canada}
\affiliation{Institute for Quantum Science and Technology, University of Calgary, Calgary, AB T2N 1N4, Canada}

\author{Salini Karuvade}
\affiliation{Department of Physics \& Astronomy, University of Calgary, Calgary, AB T2N 1N4, Canada}
\affiliation{Institute for Quantum Science and Technology, University of Calgary, Calgary, AB T2N 1N4, Canada}

\author{Carlo Maria Scandolo}
\email{carlomaria.scandolo@ucalgary.ca}
\affiliation{Department of Mathematics \& Statistics, University of Calgary, Calgary, AB T2N 1N4, Canada}
\affiliation{Institute for Quantum Science and Technology, University of Calgary, Calgary, AB T2N 1N4, Canada}

\begin{abstract}
PT-symmetric quantum theory was originally proposed with the aim of extending standard quantum theory
by relaxing the Hermiticity constraint on Hamiltonians.
However, no such extension has been formulated that consistently
describes states, transformations, measurements and composition, 
which is a requirement for
any physical theory. 
We aim to answer the question of whether a consistent physical theory 
with PT-symmetric observables extends standard quantum theory.
We answer this question within the framework of general probabilistic theories,
which is the most general framework
for physical theories.
We construct the set of states of a system that result from 
imposing PT-symmetry on the set of observables, and show that
the resulting theory allows only one trivial state.
We next consider the constraint of quasi-Hermiticity on  observables, 
which guarantees the unitarity of evolution under a Hamiltonian with unbroken PT-symmetry.
We show that 
such a system is equivalent to a standard quantum system. 
Finally, we show that if all observables are quasi-Hermitian as well as PT-symmetric, 
then the system is equivalent to a
real quantum system.
Thus our results show that neither PT-symmetry nor quasi-Hermiticity constraints are 
sufficient to extend standard quantum theory consistently.
\end{abstract}

\maketitle

\section{Introduction}
In standard quantum theory, 
the observables of a system are constrained to be Hermitian operators in order 
to guarantee real and well-defined expectation values.
PT-symmetric quantum theory was originally proposed with the aim of extending standard quantum theory
by relaxing the assumption of Hermiticity on the Hamiltonian \cite{BB98,BBJ02,BBJ03,Ben05}.
In particular, the development of PT-symmetric quantum theory was motivated by the observation that 
a Hamiltonian possesses real energy values if 
the Hamiltonian and its eigenvectors are invariant under an antilinear PT-symmetry.
If, on the one hand PT-symmetric quantum theory has 
witnessed numerous theoretical~\cite{BBJM07,Mos09,MECM08,BBC+13,VMA16,KS20,YV21,KAS22} and experimental advances in the recent years~\cite{RME+10,SLZ+11,BDG+12,POS+14,ZZS+16,XZB+17,EMK+18,WLG+19,ZLW+20,ZWX+20}, 
on the other hand, an operational foundation for PT-symmetric 
quantum theory that consistently extends standard quantum theory has not been formulated.
The absence of such a consistent extension has led to disputable proposed
applications of PT-symmetry that
contradict established information-theoretic principles including the no-signalling principle,
faster-than-Hermitian evolution of quantum states,
and the invariance of entanglement under local operations~\cite{Pat14,CCC14,YHFL14,Cro15}.
In this article, 
we answer the question of whether a consistent physical
theory with PT-symmetric observables that extends standard quantum theory can be found. 
We answer this question in the negative using the framework of general probabilistic theories (GPTs)~\cite{Hardy-informational-1,Barrett,Chiribella-purification,hardy2011,Barnum-2,Janotta-Hinrichsen,Barnum2016,ScandoloPhD}.

Formulating PT-symmetric quantum theory as a self-consistent physical theory has been a long-standing research problem\footnote{Here we review only the works that deal with the consistency of 
(first-quantized) PT-symmetric quantum theory. We remark that 
 self-consistent formulations of
PT-symmetric quantum field theories have also been investigated~\cite{AEM20,AEM22},
but they are outside the scope of the present work.}.
Efforts to construct a physical theory involving PT-symmetric Hamiltonians can be
divided into two broad categories: the quasi-Hermitian formulation for unbroken PT-symmetry \cite{Mos03a,Mos10a,Mos10b,Zno15,Mos18,KAS22} and
the Krein-space formulation \cite{Jap02,Tan06a,Tan06b,Mos06,Zel11,CCID17}.
The quasi-Hermitian approach shows that
a physical system with an unbroken PT-symmetric Hamiltonian is equivalent to a
standard quantum system, and therefore this approach does not extend standard quantum mechanics.
The Krein space formulation attempts to extend standard quantum theory to 
include PT-symmetric quantum theory, 
but this approach has not succeeded in formulating
a self-consistent physical theory. We next discuss both these approaches 
and their shortcomings.

A PT-symmetric Hamiltonian that is not Hermitian
leads to non-unitary time evolution, and, consequently,
the system violates the conservation of total probability~\cite{Zno01,BBJ02,Ben05}.
 This problem was initially circumvented by introducing a new inner product
on the Hilbert space, referred to as ``CPT inner product'', with respect to which the PT-symmetric Hamiltonian
is Hermitian~\cite{BBJ02}. Note that the CPT inner product depends on the Hamiltonian of the system as well as the PT-symmetry.
The evolution of the system is then unitary with respect to this new inner product, and therefore
conserves probability.
This  approach motivated the further development of quasi-Hermitian quantum theory, 
whereby one introduces a different inner product from the standard one.
Such a different inner product had previously been used to study systems modelled by effective non-Hermitian Hamiltonians~\cite{SGH92},
but its application to the search for extensions of quantum mechanics was driven by
the field of PT-symmetry.
According to quasi-Hermitian quantum theory,
a closed physical system with a
quasi-Hermitian Hamiltonian 
can generate unitary time evolution if the system dynamics is
considered on a modified Hilbert space with a 
Hamiltonian-dependent inner product~\cite{Mos03a,Mos10a,Mos10b,Zno15,Mos18,ZWG19,JMCN19}.
Operational foundations of quasi-Hermitian quantum theory 
and the equivalence of the resulting theory to standard quantum theory 
 follow from the fact that quasi-Hermitian observables
form a C$^*$-algebra that is isomorphic to the C$^*$-algebra of
Hermitian observables~\cite{KAS22}.
Every unbroken PT-symmetric Hamiltonian is quasi-Hermitian with respect to a suitably modified inner product~\cite{Mos03a},
and therefore, a physical system with an unbroken PT-symmetric Hamiltonian
is equivalent to a standard quantum system.

Thus, the quasi-Hermitian approach to unbroken PT-symmetry does not extend standard quantum theory. 
In fact, the equivalence of unbroken PT-symmetric systems to standard quantum systems 
has been used to successfully refute
the claims involving applications of PT-symmetry that contradicted information-theoretic principles~\cite{Mos03a,GHS08,Mos10a,Mos10b,Zno15,Mos18,ZWG19}. 
In addition to not providing an extension to standard quantum mechanics, 
it is to be noted that the quasi-Hermitian approach bypasses the original idea of 
introducing PT-symmetry.
The allowed set of observables in quasi-Hermitian quantum theory 
are only required to be Hermitian with respect to
the modified inner product, and they do not have to satisfy the PT-symmetry, if any, of the system Hamiltonian.
Therefore, quasi-Hermitian quantum theory is constructed by
actually replacing the constraint of PT-symmetry with that of quasi-Hermiticity.

The Krein space approach aims to establish a self-consistent formulation of 
PT-symmetric quantum theory, for both broken and unbroken PT-symmetric Hamiltonians.
In this formulation, the set of allowed states in PT-symmetric quantum theory form a Krein space,
which is a vector space
equipped with an indefinite inner product derived from PT-symmetry~\cite{Jap02,Tan06a,Tan06b,Mos06,Zel11,CCID17}.
The indefiniteness of the inner product imposes further restrictions on the theory,
going beyond the original requirement of PT-symmetric invariance for the Hamiltonian, 
such as a superselection rule prohibiting superposition of states from certain subspaces
and the calculation of measurement probabilities being restricted to these subspaces. 
Despite these restrictions, whether the resulting theory is self-consistent remains an open question. 
As an operational interpretation of this theory has not been investigated yet,
the question of whether it extends standard quantum mechanics cannot be answered at this stage.
Moreover, the Krein-space formulation is only applicable to 
PT-symmetric Hamiltonians that are Schr\"{o}dinger operators, and therefore does not encompass
finite-dimensional systems.

 In this article, we first show that 
the only consistent way to construct PT-symmetric quantum theory
with unbroken PT-symmetric observables, 
without any Hermiticity or quasi-Hermiticity constraint, 
is by assigning a single, trivial state with every physical system.
This result shows that PT-symmetry alone is too weak a constraint on the
set of observables to construct a non-trivial physical theory. 
We therefore investigate
the consequences of imposing different constraints related to PT-symmetry on the set of observables.
A prime candidate for such a constraint is quasi-Hermiticity, which 
has been studied in the context of unbroken PT-symmetry, as mentioned above. 
We show that if quasi-Hermiticity is the only constraint on the observables, 
then the resulting system is mathematically equivalent to a standard quantum system,
thereby recovering the results of Refs.~\cite{Mos10a,Mos10b,Zno15,Mos18,ZWG19,KAS22}
in a rigorous operational framework. However, we eliminate 
the assumption that pure states in the new theory constitute a Hilbert space, 
as done, instead, in the existing literature.
Finally, we consider the
setting in which \emph{all} observables are quasi-Hermitian as well as PT-symmetric,
and show that the resulting system is equivalent to a
real quantum system \cite{Wootters-real,Hardy-real,Baez,Hickey,Barnum2020composites,Scandolo-real1,Scandolo-real2,Acin-complex}.
Our results show that neither PT-symmetry nor quasi-Hermiticity constraints are 
sufficient to extend standard quantum theory consistently.

Our results are derived by applying 
the foundational and rigorous framework of GPTs \cite{Hardy-informational-1,Barrett,Chiribella-purification,hardy2011,Barnum-2,Janotta-Hinrichsen,Barnum2016,ScandoloPhD} to 
non-Hermitian quantum theory. 
GPTs are a framework where one only assumes that the theory is
probabilistic; as such, GPTs can accommodate theories beyond quantum theory. 
This framework is routinely
applied to studying foundational aspects of quantum theory and other theories \cite{Chiribella-informational,Hardy-informational-2,Brukner,Masanes-physical-derivation,Barnum-interference,chiribella2016quantum,ScandoloPhD,Reconstruction}.
In particular, it is possible to define probabilities associated with the measurement of physical observables via a duality between states and basic effects of a theory.
Here we show how different constraints imposed on the observables of the system, 
such as PT-symmetry and quasi-Hermiticity, 
 allow us to characterize the set of valid states of the system. This set is then compared to the one of standard quantum theory to check if PT-symmetric or quasi-Hermitian constraints provide an actual extension of quantum theory.

The organization of this article is as follows.  
In Section~\ref{sec:background} we present the background material on PT-symmetric Hamiltonians,
quasi-Hermitian Hamiltonians 
and GPTs. 
Section~\ref{sec:GPT-PT} contains the GPT-treatment of a theory where the only constraint on observables and effects is unbroken PT-symmetry. 
In Section~\ref{sec:quasi-hermitian}, we discuss the mathematical equivalence of
quasi-Hermitian quantum systems with standard quantum systems. 
Finally, Section~\ref{sec:PT+quasi} contains a discussion of 
how real quantum systems emerge from the combination of PT-symmetric and quasi-Hermitian constraints. 
Conclusions are drawn in Section~\ref{sec:conclusions}.

\section{Notation and background}
\label{sec:background}
In the following discussion, we denote a system by A.
In subsections~\ref{subsec:PT-symmetry} and~\ref{subsec:quasi-hermitian},
we associate the Hilbert space $\mathscr{H}$ with A.
The inner product defined in $\mathscr{H}$ will be denoted by $\left\langle\bullet \middle|\bullet \right\rangle$. For convenience, we will assume $\mathscr{H}$ to be finite-dimensional, as this is the usual scenario in which the GPT framework is applied. The set $\mathcal{B}\left( \mathscr{H}\right)$ comprises the bounded linear operators acting
on $\mathscr{H}$.

\subsection{PT-symmetric quantum theory}\label{subsec:PT-symmetry}

In this section, we review the basics of PT-symmetric quantum theory.
PT-symmetric quantum theory replaces the Hermiticity constraint on observables
in standard quantum mechanics by the physically-motivated constraint of invariance under PT-symmetry.
The operator PT acting on the Hilbert space $\mathscr{H}$ is assumed to be the
composition of P, a linear operator and T, an antilinear operator, 
such that their combined action is an antiunitary involution on $\mathscr{H}$.

\begin{defn}\label{def:antiinvo}
A (linear or antilinear) operator $M$ acting on the Hilbert space $\mathscr{H}$ is an \emph{involution}
if it satisfies $M^2=\mathds{1}.$
\end{defn}

Any antiunitary operator that is also an involution is a valid PT-symmetry, and 
can be used to construct an instance of PT-symmetric quantum theory. For a given choice of PT-symmetry,
the time evolution of a physical system A is described by a 
Hamiltonian $ H_{\rm PT}\in\mathcal{B}\left(\mathscr{H}\right)$
that is not necessarily Hermitian, but it is invariant under the
action of PT. That is,
\begin{equation}\label{eq:PTcommute}
   \left[ H_{\rm PT},\mathrm{PT}\right] = 0.
\end{equation}

\noindent
Even if $ H_{\rm PT}$ is diagonalizable,
Eq.~\eqref{eq:PTcommute} is not sufficient to guarantee that $H_{\rm PT}$ and $\mathrm{PT}$ share an eigenbasis,
as PT is an antilinear operator.
This brings us to the definition of an unbroken PT-symmetric Hamiltonian.
\begin{defn}[Adapted from~\cite{BBJ02}]\label{def:unbroken}
A Hamiltonian $ H_{\rm PT}$ is called an \emph{unbroken PT-symmetric Hamiltonian},
or in other words $ H_{\rm PT}$ is said to possess \emph{unbroken PT-symmetry}, 
if $ H_{\rm PT}$ is diagonalizable, and all the eigenvectors of $ H_{\rm PT}$ are invariant under the action of PT. 
\end{defn}

As a consequence of this definition, an unbroken PT-symmetric Hamiltonian possesses real spectrum.
If not all eigenvectors of $ H_{\rm PT}$ are invariant under the action of PT, as is the case for
Hamiltonians with broken PT-symmetry, then the spectrum of $H_{\rm PT}$
could consist of complex eigenvalues that arise in complex conjugate pairs (see e.g.\ Ref.~\cite{Mos02}). 
Instead, non-diagonalizable PT-symmetric Hamiltonians describe exceptional points, characterized by coalescence of
one or multiple pairs of eigenvalues and eigenvectors~\cite{DDT01}.

\begin{exmp}
Consider $\mathrm{P} = \sigma_x$ and $\mathrm{T} = \kappa$, where $\sigma_x$ is the Pauli $x$ matrix, and $\kappa$ denotes  
complex conjugation. Then the Hamiltonian 
\begin{equation}
\label{2x2}
    H_{\rm PT} = \begin{pmatrix}r\mathrm{e}^{\mathrm{i}\theta} & s \\ s & r\mathrm{e}^{-\mathrm{i}\theta}\end{pmatrix},
    \quad r,s,\theta \in \mathds{R}
\end{equation}
commutes with the the product $\mathrm{PT}$, i.e.\ $H_{\rm PT}$ is PT-symmetric~\cite{BBJ02}.
The eigenvalues of $H_{\rm PT}$ are $\lambda_{\pm} = r\cos\theta\pm\sqrt{s^2-r^2\sin^2\theta}$,
which are real for $s^2\geq r^2\sin^2\theta$. Furthermore, note that if $s^2 > r^2\sin^2\theta$, $H_{\rm PT}$
is diagonalizable, and the corresponding
eigenvectors of $H_{\rm PT}$
can be chosen to be 
\begin{align}\label{eq:eigenvector examples}
    \ket{\lambda_+} &= \frac{1}{\sqrt{2\cos\alpha}}\begin{pmatrix}\mathrm{e}^{\mathrm{i}\alpha/2}\\ \mathrm{e}^{-\mathrm{i}\alpha/2}\end{pmatrix}\\
    \ket{\lambda_-} &= \frac{\mathrm{i}}{\sqrt{2\cos\alpha}}\begin{pmatrix}\mathrm{e}^{-\mathrm{i}\alpha/2}\\ -\mathrm{e}^{\mathrm{i}\alpha/2}\end{pmatrix},\label{eq:eigenvector examples2}
\end{align}
where $\alpha\in\left(-\frac{\pi}{2},\frac{\pi}{2}\right)$ is such that $\sin\alpha := \left(r/s\right)\sin\theta$ 
(note that such an $\alpha$ exists for the parameter values satisfying $s^2 > r^2\sin^2\theta$).
The eigenvectors $\left\{\ket{\lambda_+},\ket{\lambda_-}\right\}$ 
are easily verified to be eigenvectors of the antilinear operator $\mathrm{PT}$. 
Hence, $H_{\rm PT}$ 
in this parameter regime displays unbroken PT-symmetry. 
For $s = \pm r\sin\theta$, the Hamiltonian $H_{\rm PT}$ becomes non-diagonalizable, and this parameter regime is called exceptional
point. For $s^2<r^2\sin^2\theta$, the Hamiltonian
$H_{\rm PT}$ has complex eigenvalues, and thus enters a broken PT-symmetric phase.
\end{exmp}

\subsection{Quasi-Hermitian quantum theory}\label{subsec:quasi-hermitian}
In this section, we discuss the fundamental concepts in  quasi-Hermitian quantum theory
and the relation of the theory to PT-symmetric Hamiltonians.
We begin with the definition of a quasi-Hermitian operator, 
which forms the basis of this theory.
\begin{defn}~\label{def:qH}
Let $\eta\in \mathcal{B}\left(\mathscr{H}\right)$
be a positive definite operator. 
An operator $M\in\mathcal{B}\left(\mathscr{H}\right)$ is \emph{quasi-Hermitian} with respect to the \emph{metric operator} $\eta$, or $\eta$-\emph{Hermitian},
if it satisfies the condition 
\begin{equation}\label{eq:qH}
    \eta M \eta^{-1} =M^\dagger.
\end{equation}
\end{defn}

In quasi-Hermitian quantum theory, the dynamics of a system A is generated by 
a non-Hermitian Hamiltonian $H_{\rm QH}\in\mathcal{B}\left(\mathscr{H}\right)$
that is quasi-Hermitian~\cite{Mos02}.
Consequently, the evolution generated by $H_{\rm QH}$ is not unitary.
In particular, a quasi-Hermitian Hamiltonian generates a quasi-unitary evolution in $\mathscr{H}$.
We define a quasi-unitary evolution below.

\begin{defn}~\label{def:qU}
An operator $M\in\mathcal{B}\left(\mathscr{H}\right)$ is \emph{quasi-unitary} or $\eta$-\emph{unitary}
if it satisfies the condition 
\begin{equation*}
    M^\dagger \eta M =\eta.
\end{equation*}
\end{defn}
\noindent Clearly, a closed system A undergoing a quasi-unitary evolution
violates the conservation of probability in $\mathscr{H}$.
However, in quasi-Hermitian quantum theory,
unitarity of evolution is restored by modifying the system Hilbert space to
 $\mathscr{H}_\eta$,
consisting of the underlying vector space $\mathscr{V}$ with a modified inner product 
$\left\langle\bullet\middle|\bullet\right\rangle_\eta$ given by 
\begin{equation}
\label{eq:etaIP}
   \left\langle\phi\middle|\psi\right\rangle_\eta := \left\langle\phi\middle|\eta\middle|\psi\right\rangle,\quad\forall\ket{\phi},\ket{\psi}\in\mathscr{V}.
\end{equation}
We refer to this modified inner product as the \emph{$\eta$-inner product}. In this setting, $\eta$-unitary operators preserve the $\eta$-inner product.
It is easy to verify that $H_{\rm QH}$ acts as a Hermitian operator on $\mathscr{H}_\eta$~\cite{SGH92,Mos02}.
Any valid observable in A is also required to satisfy Eq.~\eqref{eq:qH},
or equivalently be represented by a Hermitian operator
in $\mathcal{B}\left(\mathscr{H}_\eta\right)$, in this theory.
Furthermore, the isomorphism between $\mathscr{H}$ and $\mathscr{H}_\eta$ implies that
a closed quasi-Hermitian system with Hamiltonian $H_{\rm QH}$ is mathematically equivalent to a 
standard quantum system, provided the dynamics of the former is described in $\mathscr{H}_\eta$. 
In other words, Hamiltonians that are  $\eta$-Hermitian in $\mathscr{H}$ generate unitary evolution in~$\mathscr{H}_\eta$.

Eq.~\eqref{eq:qH} is a necessary and sufficient condition for any diagonalizable
operator in $\mathcal{B}\left(\mathscr{H}\right)$ to possess a real  spectrum~\cite{Mos03a}.
Consequently, any Hamiltonian $H_{\rm PT}$ with unbroken PT-symmetry
satisfies Eq.~\eqref{eq:qH} for some metric operator $\eta$.

Below we give an example of such a metric operator for an unbroken PT-symmetric Hamiltonian.
\begin{exmp}
The Hamiltonian $H_{\rm PT}$ in Eq.\,\eqref{2x2}
is Hermitian with respect to the inner product
\begin{equation}
\label{CPT}
    \braket{\psi|\phi}_{\rm CPT} = \left(\mathrm{CPT}\ket{\psi}\right)\cdot \ket{\phi},
\end{equation}
where, for $\alpha\in \left(-\frac{\pi}{2},\frac{\pi}{2}\right)$,
\[
    \mathrm{C} = \frac{1}{\cos\alpha}\begin{pmatrix}\mathrm{i}\sin\alpha & 1 \\ 1 & -\mathrm{i}\sin\alpha
    \end{pmatrix}
\]
is the ``Charge'' operator~\cite{BBJ02}, where $\alpha$ is defined as in Eqs.~\eqref{eq:eigenvector examples} and \eqref{eq:eigenvector examples2}, and
$a\cdot b = \sum_{j}a_jb_j$ denotes the dot product. 
The easiest approach to verify that $H_{\rm PT}$ is Hermitian with respect to
the inner product in Eq.~\eqref{CPT} is by considering the action of $\mathrm{C}$ and $\mathrm{PT}$
on the eigenvectors of $H_{\rm PT}$. Using the explicit form of the eigenvectors of 
$H_{\rm PT}$ in Eqs.~\eqref{eq:eigenvector examples} and \eqref{eq:eigenvector examples2}, it is straightforward to see that  
$\mathrm{C}\ket{\lambda_{\pm}} = \pm\ket{\lambda_{\pm}}$. We already know that
$\mathrm{PT}\ket{\lambda_{\pm}} = \ket{\lambda_{\pm}}$.
The normalization factor $1/{\sqrt{2\cos\alpha}}$ in Eqs.~\eqref{eq:eigenvector examples} and \eqref{eq:eigenvector examples2}
ensures that 
$\mathrm{C}\mathrm{PT}\ket{\lambda_{\pm}}\cdot\ket{\lambda_{\pm}} = 1$ and $\mathrm{C}\mathrm{PT}\ket{\lambda_{\pm}}\cdot \ket{\lambda_{\mp}} = 0$.
Thus, the eigenvectors of $H_{\rm PT}$ are orthonormal with respect to the inner product in Eq.~\eqref{CPT}.
Having already proven the reality of the eigenvalues of $H_{\rm PT}$, 
we conclude that $H_{\rm PT}$ is Hermitian with respect to the inner product in Eq.~\eqref{CPT}.
Finally, we can once again use the action of $\mathrm{C}$ and $\mathrm{PT}$ to show that $\mathrm{C}\mathrm{PT} = \mathrm{PT}\mathrm{C}=\mathrm{TPC}$,
where we also 
used $\mathrm{PT}=\mathrm{TP}$ in the latter step. Now, the inner product in Eq.~\eqref{CPT} can be re-expressed as
\begin{equation}
    \braket{\psi|\phi}_{\rm CPT} = \braket{\psi|\mathrm{PC}|\phi},
\end{equation}
which is equivalent to $\eta$-inner product 
with $\eta = \mathrm{PC}$. 
\end{exmp}

In contrast to the requirement of PT-symmetry on the Hamiltonian, 
the observables of the quasi-Hermitian system associated with $H_{\rm PT}$ are not 
traditionally required to be invariant under the same PT-symmetry.
Instead, observables are required to be quasi-Hermitian with respect to the CPT inner product,
or to another inner product which makes $H_{\rm PT}$ Hermitian\footnote{Ref.~\cite{BBJ02}
erroneously posited that the observables of a PT-symmetric system should satisfy CPT-symmetry, but
later clarified in the erratum that this condition should be replaced by $\eta$-Hermiticity
with $\eta = \mathrm{PC}$.}. This contrast in requirements
on the Hamiltonian and other observables is deemed to be necessary to maintain consistency,
so as to ensure reality of eigenvalues of the observables 
and unitarity of the evolution generated by the Hamiltonian. 
However, the reality of the eigenvalues of observables and unitarity of the evolution
are not fundamental to the consistency of every physical theory, although they are 
integral to the consistency of standard quantum mechanics.
In fact, a rigorous operational assessment
of the consequences of requiring observables to be PT-symmetric (and not quasi-Hermitian)
has never been carried out in the literature. This is exactly the starting point of our 
analysis in \S \ref{sec:GPT-PT}.

\subsection{General probabilistic theories}\label{subsec:GPTs}
In this section, we review the basic structure of general probabilistic theories 
(GPTs)~\cite{Hardy-informational-1,Barrett,Chiribella-purification,hardy2011,Barnum-2,Janotta-Hinrichsen,Barnum2016,ScandoloPhD}.
There are different ways to introduce GPTs; here we opt for a minimalist treatment, inspired by Ref.~\cite{Randall-Foulis1}, which focuses on states and effects, the objects of interest of our analysis. Here effects indicate the mathematical objects associated with the various outcomes of the measurement of physical observables (possibly even generalized ones, such as those associated with POVMs in quantum theory).

For every system $\mathrm{A}$, we identify a set of basic effects $\mathsf{X}\left(\mathrm{A}\right)$, and a set of basic measurements $\mathsf{M}_\mathsf{X}\left(\mathrm{A}\right)$, which are particular collections of basic effects. We can think of 
any basic measurement to be associated with a certain physical observable. It is assumed that the basic measurements in  $\mathsf{M}_\mathsf{X}\left(\mathrm{A}\right)$ provide a covering of $\mathsf{X}$. A state $\mu$ of the system is a \emph{probability weight}, i.e.\ a function $\mu:\mathsf{X}\to \left[0,1\right]$ such that, for every basic measurement $\boldsymbol{m}\in \mathsf{M}_\mathsf{X}\left(\mathrm{A}\right)$, we have $\sum_{E\in \boldsymbol{m}}\mu \left(E\right) = 1$. In simpler words, a state assigns a probability to every measurement outcome: $\mu\left(E\right)$ is the probability of obtaining the outcome associated with $E$ if the state is $\mu$. The set of states of system $\mathrm{A}$, denoted by $\mathsf{St}\left(\mathrm{A}\right)$, is a convex set, because any convex combination of probability weights is still a probability weight.

Since states are real-valued functions, we can define linear combinations
of them with real coefficients: if $a,b\in\mathds{R}$ and $\mu,\nu\in\mathsf{St}\left(\mathrm{A}\right)$,
then $a\mu+b\nu$ is defined as
\[
\left(a\mu+b\nu\right)\left(E\right):=a\mu\left(E\right)+b\nu\left(E\right),
\]
for every $E\in\mathsf{X}\left(\mathrm{A}\right)$. In this way, states
span a real vector space, denoted as $\mathsf{St}_{\mathds{R}}\left(\mathrm{A}\right)$.
Hereafter, we assume that $\mathsf{St}_{\mathds{R}}\left(\mathrm{A}\right)$
is finite-dimensional. Note that basic effects can be regarded as particular linear
functionals on $\mathsf{St}_{\mathds{R}}\left(\mathrm{A}\right)$:
if $E\in\mathsf{X}\left(\mathrm{A}\right)$, then $E\left(\mu\right):=\mu\left(E\right)$,
where $\mu\in\mathsf{St}\left(\mathrm{A}\right)$. Similarly, one
can consider the real vector space spanned by basic effects, denoted
by $\mathsf{Eff}_{\mathds{R}}\left(\mathrm{A}\right)$. Note that
$\mathsf{St}_{\mathds{R}}\left(\mathrm{A}\right)$ is the dual space
of $\mathsf{Eff}_{\mathds{R}}\left(\mathrm{A}\right)$. Within  $\mathsf{Eff}_{\mathds{R}}\left(\mathrm{A}\right)$ one identifies the set of effects, $\mathsf{Eff}\left(\mathrm{A}\right)$, which comprises all elements of $\mathsf{Eff}_{\mathds{R}}\left(\mathrm{A}\right)$ that can arise in a physical measurement on the system, even if they are not basic effects. For physical consistency, effects in $\mathsf{Eff}\left(\mathrm{A}\right)$ must be such that states are still probability weights on them. Certain collections of effects that are not necessarily basic effects make up more general measurements than basic measurements. Yet the property $\sum_{E\in \boldsymbol{m}}\mu \left(E\right) = 1$ for a state $\mu$ still holds even when $\boldsymbol{m}$ is a general type of measurement. In this sense, every effect in $\mathsf{Eff}\left(\mathrm{A}\right)$ must be part of some measurement.

\begin{exmp}
In quantum theory, for every system, basic effects can be taken
to be rank-1 orthogonal projectors, in which case basic measurements
are all the collections of rank-1 projectors that sum to the identity.
Basic effects span the vector space of Hermitian matrices. 
The set of effects is the set of POVM elements, namely operators $E$ such that $\bm{0}\leq E \leq \mathds{1}$, and measurements are all POVMs.
 With the formalism
presented above, states are particular linear functionals on the vector
space spanned by basic effects. According to a theorem by Gleason \cite{gleason},
they are of the form $\mathrm{tr}\:\rho\bullet$, where $\rho$ is
any positive semidefinite matrix with trace 1 (density matrix), and
$\bullet$ is a placeholder for a basic effect. In
other words, there is a bijection between quantum states and density
matrices. This is the reason why quantum states are commonly defined
as density matrices, forgetting their nature as linear functionals.
\end{exmp}

Two states (resp.\ two effects) are equal if their action on all
effects (resp.\ states) is the same. In this way, it is possible
to show that the effects in all basic measurements sum to the same
linear functional $u$, known as \emph{unit effect} or \emph{deterministic
effect}. Indeed, if $\boldsymbol{m}$ and $\boldsymbol{m}'$ are two
basic measurements, and $\mu$ is a state,
\[
\sum_{E\in\boldsymbol{m}}E\left(\mu\right)=\sum_{E\in\boldsymbol{m}}\mu\left(E\right)=1=\sum_{F\in\boldsymbol{m}'}\mu\left(F\right)=\sum_{F\in\boldsymbol{m}'}F\left(\mu\right).
\]
Thus, $\sum_{E\in\boldsymbol{m}}E=\sum_{F\in\boldsymbol{m}'}F$. This fact guarantees that the principle of Causality is in force,
so the theory cannot have signalling in space and time \cite{Chiribella-purification}.
\begin{exmp}
In quantum theory, the unit effect is the identity, as all rank-1 projectors in a basic measurement sum to the identity.
\end{exmp}

In this framework, in any theory we can define a physical observable $O$ mathematically, starting from the basic measurement $\boldsymbol{m}=\left\{E_1,\dots,E_s\right\}$ associated with it. Let $\left\{\lambda_1,\dots,\lambda_s\right\}$ be the (possibly equal) values of the observable that can be found after a measurement, where $\lambda_j$ is the value associated with the $j$th outcome, i.e.\ with the effect $E_j$. Then we can represent $O$ as a linear combination of the basic effects  in $\boldsymbol{m}$, where the coefficients are its values (cf.\ \cite{ScandoloPhD}):
\begin{equation*}
    O=\sum_{j=1}^s \lambda_j E_j.
\end{equation*}
Therefore, from the mathematical point of view, observables are particular elements of $\mathsf{Eff}_{\mathds{R}}\left(\mathrm{A}\right)$. In this way, it is possible to define the expectation value of the observable $O$ on the state $\mu$ as
\begin{equation}\label{eq:expectation value}
    \left\langle O\right\rangle_\mu=\mu\left(O\right)=\sum_{j=1}^s \lambda_j \mu\left(E_j\right).
\end{equation}
Given that $\mu\left(E_j\right)$ represents the probability of obtaining $E_j$, i.e.\ of obtaining the value $\lambda_j$, the meaning of Eq.~\eqref{eq:expectation value} as an expectation value is clear. This also shows the tight relationship between observables and effects, which implies that any constraint imposed on observables of a theory can be viewed directly as a constraint on the effects of the theory. We will see an example of this in Proposition~\ref{prop:PTspectral}.
\begin{exmp}
In quantum theory, observables are indeed linear combinations of basic effects. This can be seen as a consequence of the fact that observables in finite-dimensional systems are represented by Hermitian matrices: in this way, every observable is diagonalizable, and therefore it can be written as a linear combination of its spectral projectors with the coefficients being its eigenvalues, i.e.\ the values that can be found in a measurement. In turn, every spectral projector can be written as a sum of rank-1 projectors, so every quantum observable can be written as a linear combination of basic effects that make up a basic measurement (they sum to the identity).
\end{exmp}

We end this section by defining the meaning of equivalence between two physical systems. 
\begin{defn}\label{def:equivalence}
Let $\mathrm{A}$ and $\mathrm{B}$ be two physical systems. 
We say that $\mathrm{A}$ is \emph{equivalent} to $\mathrm{B}$ if there exists a linear bijection
$\mathcal{T}:\mathsf{Eff}\left(\mathrm{A}\right) \to \mathsf{Eff}\left(\mathrm{B}\right)$ such that $\mathcal{T}\left(u_{\mathrm{A}}\right)= u_{\mathrm{B}}$, where $u$ denotes the unit effect.
\end{defn}
Note that such a $\mathcal{T}$ can be extended by linearity to become an isomorphism of the vector spaces $\mathsf{Eff}_{\mathds{R}}\left(\mathrm{A}\right)$ and $\mathsf{Eff}_{\mathds{R}}\left(\mathrm{B}\right)$, because such vector spaces are spanned by effects. The existence of a linear bijection between effect spaces implies the existence of a linear bijection between state spaces, which justifies the equivalence at an even stronger level.
\begin{lem}\label{lem:duality}
If $\mathrm{A}$ and $\mathrm{B}$ are equivalent, then there exists a linear bijection $\mathcal{T}':\mathsf{St}\left(\mathrm{B}\right) \to \mathsf{St}\left(\mathrm{A}\right)$.
\end{lem}
Again, such a $\mathcal{T}'$ can be extended by linearity to become an isomorphism of the vector spaces $\mathsf{St}_{\mathds{R}}\left(\mathrm{B}\right)$ and $\mathsf{St}_{\mathds{R}}\left(\mathrm{A}\right)$.
\begin{proof}
Note that if $\mathrm{A}$ and $\mathrm{B}$ are equivalent, we can construct $\mathcal{T}':\mathsf{St}_{\mathds{R}}\left(\mathrm{B}\right)\to \mathsf{St}_{\mathds{R}}\left(\mathrm{A}\right)$ as the dual map of the isomorphism $\mathcal{T}:\mathsf{Eff}_{\mathds{R}}\left(\mathrm{A}\right) \to \mathsf{Eff}_{\mathds{R}}\left(\mathrm{B}\right)$ introduced in Definition~\ref{def:equivalence}. With such a construction, $\mathcal{T}':\mathsf{St}_{\mathds{R}}\left(\mathrm{B}\right) \to \mathsf{St}_{\mathds{R}}\left(\mathrm{A}\right)$ is defined as $\nu \mapsto \mu$ such that $\mu\left(E\right):=\nu\left(\mathcal{T}\left(E\right)\right)$, for every $E\in\mathsf{Eff}\left(\mathrm{A}\right)$. It is known from linear algebra that $\mathcal{T}':\mathsf{St}_{\mathds{R}}\left(\mathrm{B}\right) \to \mathsf{St}_{\mathds{R}}\left(\mathrm{A}\right)$ is also an isomorphism.

Now, to prove the lemma, it is enough we prove that $\mathcal{T}'\left(\mathsf{St}\left(\mathrm{B}\right)\right)= \mathsf{St}\left(\mathrm{A}\right)$. We first show that $\mathcal{T}'\left(\mathsf{St}\left(\mathrm{B}\right)\right)\subseteq \mathsf{St}\left(\mathrm{A}\right)$. Observe that
$\mu\left(E\right)=\nu\left(\mathcal{T}\left(E\right)\right)\geq 0$ because $\mathcal{T}\left(E\right)\in \mathsf{Eff}\left(\mathrm{B}\right)$ and $\nu\in \mathsf{St}\left(\mathrm{B}\right)$. Moreover, \[\mu\left(u\right)=\nu\left(\mathcal{T}\left(u\right)\right)=\nu\left(u\right)=1,
\]
because $\mathcal{T}\left(u\right)=u$. The fact that $\mu\left(u\right)=1$ also ensures that $\mu\left(E\right)\leq 1$, for every effect $E\in\mathsf{Eff}\left(\mathrm{A}\right)$. Then $\mu\in\mathsf{St}\left(\mathrm{A}\right)$. This shows that $\mathcal{T}'\left(\mathsf{St}\left(\mathrm{B}\right)\right)\subseteq \mathsf{St}\left(\mathrm{A}\right)$.

To show the other inclusion, let us consider a state $\mu\in\mathsf{St}\left(\mathrm{A}\right)$, and let us show we can find a $\nu\in \mathsf{St}\left(\mathrm{B}\right)$ such that $\mu\left(E\right)=\nu\left(\mathcal{T}\left(E\right)\right)$ for all $E\in\mathsf{Eff}\left(\mathrm{A}\right)$. To this end, it is enough to take the state $\nu\in \mathsf{St}\left(\mathrm{B}\right)$ such that $\nu\left(F\right):=\mu\left(\mathcal{T}^{-1}\left(F\right)\right)$ for all $F\in\mathsf{Eff}\left(\mathrm{B}\right)$. Now, by hypothesis $\mathcal{T}^{-1}\left(F\right)\in\mathsf{Eff}\left(\mathrm{A}\right)$, so the definition of $\nu$ is well posed. Then, for any $E\in\mathsf{Eff}\left(\mathrm{A}\right)$, we have
\[
\mu\left(E\right)=\nu\left(\mathcal{T}\left(E\right)\right):=\mu\left(\mathcal{T}^{-1}\left(\mathcal{T}\left(E\right)\right)\right)\equiv \mu\left(E\right).
\]
This shows that $\mathcal{T}'\left(\mathsf{St}\left(\mathrm{B}\right)\right)\supseteq \mathsf{St}\left(\mathrm{A}\right)$, from which $\mathcal{T}'\left(\mathsf{St}\left(\mathrm{B}\right)\right)= \mathsf{St}\left(\mathrm{A}\right)$.
\end{proof}

This concludes our review of general probabilistic theories, and next we apply this framework to 
investigate PT-symmetric quantum theory.

\section{A GPT with PT-symmetric effects}\label{sec:GPT-PT}
In this section, we derive the structure of states in the GPT defined by PT-symmetric observables (and hence effects). 
Here we assume that the observables of the system $\mathrm{A}_{\rm PT}$ under consideration are operators on a finite-dimensional complex Hilbert space
$\mathscr{H}\cong \mathds{C}^d$, where $d$ is the dimension of the space.
However, we remove the Hermiticity constraint from the set of observables (and consequently from effects), and replace it with unbroken PT-symmetry, as given in Definition~\ref{def:unbroken}. 
We denote any valid PT-symmetry discussed in Section~\ref{subsec:PT-symmetry}  by $K$,
and consequently use this notation throughout this article.
We show that, under these assumptions, the theory allows
only one state, which is associated with a multiple of the identity matrix.

Before proving that the theory we construct only allows a single state, 
we first show that for unbroken $K$-symmetric observables, the projectors in their spectral decomposition are also $K$-symmetric. 
We later use this structure to declare $K$-symmetric projectors as the
basic effects in our theory.
We begin our analysis by defining $K$-symmetric projectors.
\begin{defn}
A projector $P$ (i.e.\ an operator satisfying $P^2=P$) 
is said to be $K$\emph{-symmetric} if it commutes with the anitunitary symmetry $K$, i.e. $KP = PK$.
\end{defn}
Note that we do \emph{not} require projectors to be Hermitian (viz.\ orthogonal), but simply idempotent ($P^2=P$).

Now we are ready to state the proposition that shows that 
any unbroken $K$-symmetric operator can be expressed as a linear combination of 
$K$-symmetric, possibly non-orthogonal, projectors onto its eigenspaces.
This proposition complements  
an observation by Bender and Boettcher in Ref.~\cite{BB98} that led to the 
development of unbroken PT-symmetric quantum theory, namely that the eigenvectors of an unbroken PT-symmetric Hamiltonian
are also eigenvectors of the PT operator.

\begin{prop}
\label{prop:PTspectral}
Let $O$ be an unbroken $K$-symmetric operator on $\mathds{C}^d$ with distinct 
eigenvalues $\left\{\lambda_j\right\}_{j=1}^s$. Then 
there exist spectral projectors $\left\{P_j\right\}$ satisfying
\begin{enumerate}
    \item $O = \sum_{j} \lambda_j P_{j}$;
   \item $P_{j} P_{k} = \delta_{jk} P_j$;
    \item $\sum_{j} P_j = \mathds{1}$;
    \item each $P_{j}$ is $K$-symmetric.
\end{enumerate}
\end{prop}
\begin{proof}
The observable $O$ is diagonalizable by definition of unbroken $K$-symmetry, and its spectrum is real. 
The existence of spectral projectors $\left\{P_j\right\}$ satisfying conditions~(i)--(iii) for any diagonalizable matrix $O$ is a known fact of linear algebra. We only need to prove that these projectors satisfy condition (iv). Consider a decomposition of $\mathds{C}^d$ into a direct sum of the eigenspaces of $O$:
\begin{equation*}
    \mathds{C}^d=\mathscr{V}_1 \oplus \dots \oplus \mathscr{V}_s.
\end{equation*}
Let us start by examining the projectors whose associated eigenvalue $\lambda_{j}$
is non-zero. In this case, the projector $P_{j}$ is thus defined
\[
P_{j}\left|\psi\right\rangle =\begin{cases}
\frac{1}{\lambda_{j}}O\left|\psi\right\rangle  & \left|\psi\right\rangle \in\mathscr{V}_{j}\\
0 & \left|\psi\right\rangle \notin\mathscr{V}_{j}
\end{cases}.
\]
Therefore, $P_{j}$ is $K$-symmetric on $\mathscr{V}_{j}$ because
so is $O$. It is also $K$-symmetric outside $\mathscr{V}_{j}$ because
it behaves as the zero operator, which is trivially $K$-symmetric.
Therefore, $P_{j}$ is $K$-symmetric on all $\mathds{C}^{d}$. If
present, let us consider the projector $P_{0}$ associated with the
zero eigenvalue as the last projector. It can be written as $\mathds{1}-\sum_{P_{j}\neq P_{0}}P_{j}$.
Being a linear combination with real coefficients of $K$-symmetric operators, it is $K$-symmetric
itself.
\end{proof}
\noindent This observation provides a strong motivation to construct a GPT with effects
represented by $K$-symmetric projectors
when the Hermiticity requirement for observables is replaced by unbroken $K$-symmetry.
We now proceed to construct such a theory. 

Consider a finite dimensional 
system ${\rm A}_{K}$ where the observables are not necessarily Hermitian,
but they possess unbroken $K$-symmetry.
Thanks to Proposition~\ref{prop:PTspectral}, we can take the set of basic effects of this system to be
\begin{equation*}
    \mathsf{X}\left({\rm A}_{K} \right) = \left\{P: P^2=P, PK = KP \right\}.
\end{equation*}
Basic measurements on this system are collections of $K$-symmetric projectors 
that sum to the identity operator $\mathds{1}$. 
Note that the identity operator is also $K$-symmetric by definition, implying that 
$\left\{\mathds{1}\right\}$ is a particular example of a basic measurement ($\mathds{1}$ is the unit effect). 
Consequently, any valid state  $\rho$ in this new theory must satisfy
$\rho\left(\mathds{1}\right)=1$.

We now show that in this new theory,
 system ${\rm A}_{K}$ has only a single state. 
To prove this, we start by recalling a result of linear algebra.

\begin{lem}\label{lem:linear algebra}
For every linear functional $\nu:\mathsf{S} \to \mathds{R}$ where $\mathsf{S}$ is the complex linear span of a subset of $ M_d\left(\mathds{C}\right)$---the space of complex square matrices of order $d$---there exists a $T_\nu \in M_d\left(\mathds{C}\right)$ satisfying
\begin{equation*}
    {\rm tr}\, T_\nu E = \nu\left(E\right) \quad \forall E \in \mathsf{S}.
\end{equation*}
Furthermore, such a $T_\nu$ is unique if and only if $\mathsf{S}=M_d\left(\mathds{C}\right)$.
\end{lem}
\begin{proof}
We know that, if we have an inner product on $M_{d}\left(\mathds{C}\right)$,
all linear functionals on $M_{d}\left(\mathds{C}\right)$ (and therefore
on $\mathsf{S}$) can be obtained through that inner product. Now,
we can consider the Hilbert-Schmidt inner product on $M_{d}\left(\mathds{C}\right)$,
given by $\left(E,F\right):=\mathrm{tr}\:E^{\dagger}F$, for $E,F\in M_{d}\left(\mathds{C}\right)$.
Therefore, the action of a linear functional $\nu:\mathsf{S}\to\mathds{R}$
can be represented as
\[
\nu\left(E\right)=\mathrm{tr}\,\widetilde{T}_{\nu}^{\dagger}E,
\]
for some complex square matrix $\widetilde{T}_{\nu}$, and for any $E\in \mathsf{S}$.
To get the thesis, it is enough to define $T_{\nu}:=\widetilde{T}_{\nu}^{\dagger}$. 

If $\mathsf{S}=M_d\left(\mathds{C}\right)$, then there is a unique $T_\nu$ by virtue of the isomorphism between $M_d\left(\mathds{C}\right)$ and its dual space via the established Hilbert-Schmidt inner product. To prove the converse direction, suppose by contradiction that $\mathsf{S}$ is a proper subspace of $M_{d}\left(\mathds{C}\right)$. In this case, there is not a unique way to extend a linear functional
on $\mathsf{S}$ to the whole
$M_{d}\left(\mathds{C}\right)$. This means that we can associate more than one
square matrix of order $d$ with~$\nu$.
\end{proof}
\noindent Lemma~\ref{lem:linear algebra} implies that the states in $\mathsf{St}\left({\rm A}_{K}\right)$ can be represented by 
matrices in~$M_d\left(\mathds{C}\right)$. 

Now we focus on the special case where the PT-symmetry is simply $\kappa$, the complex conjugation operation in the canonical basis,
and show that the associated $\kappa$-symmetric system ${\rm A}_{\kappa}$ admits only a single state. 
After that, we extend this result to any general PT-symmetry $K$.
In the special case $K=\kappa$, it is easy to verify that the set of all $\kappa$-symmetric effects 
are given by real projectors on $\mathds{C}^d$: 
for any $\kappa$-symmetric projector $P \in M_d\left(\mathds{C}\right)$, we have $\kappa P \kappa = P^*$ by definition of
complex conjugation, and $\kappa P \kappa = P$ by the definition
of $\kappa$ symmetry.
The following lemma adapts Lemma~\ref{lem:linear algebra} to deal with the case of real projectors and
show that any valid state in ${\rm A}_{\kappa}$ can be represented by a real matrix.
\begin{lem}
\label{lem:dual}
Let $\mathsf{Q}$
denote the $\mathds{R}$-linear span of $\kappa$-symmetric, i.e., real projectors on $\mathds{C}^d$.
For every linear functional $\nu:
\mathsf{Q} \to \mathds{R}$, there exists a 
$\kappa$-symmetric operator $T_\nu\in M_d\left(\mathds{R}\right)$ satisfying
\[
\nu\left(Q\right) = {\rm tr}\,T_\nu Q \qquad \forall Q \in \mathsf{Q}.
\]
\end{lem}

\begin{proof}
Since here we are dealing only with real projectors,
this case can be embedded in $M_{d}\left(\mathds{R}\right)$, for which an analogous statement to Lemma~\ref{lem:linear algebra} holds. Then the matrix $T_{\nu}$ of Lemma~\ref{lem:linear algebra} can be taken to be real and therefore, $\kappa$-symmetric.
\end{proof}
Now we determine the allowed states for system ${\rm A}_{\kappa}$
with $\kappa$-symmetric effects,
which we refer to as $\kappa$-symmetric states. 
According to Section~\ref{subsec:GPTs}, we need to find linear functionals on $\mathsf{Q}$ 
that yield a number in $\left[0,1\right]$ when applied to a basic effect that is $\kappa$-symmetric, 
i.e.\ a real projector.

\begin{lem}
\label{lem:realprojectors}
For every system of a $\kappa$-symmetric theory, there exists only one  state \(\nu\), given by 
\begin{equation*}
\nu\left(Q\right)= \frac{1}{d}\,{\rm rk}\,Q,
\end{equation*}
for all  real projectors $Q$, where $d$ is the dimension of the system, and ${\rm rk}$ denotes the rank of the matrix.
\end{lem}

\begin{proof}
Let $\nu$ be a $\kappa$-symmetric state, and let $T_\nu$ be a
real operator satisfying
\[
\nu\left(Q\right) = {\rm tr}\,T_\nu Q 
\]
for every real projector $Q$ (cf.\ Lemma~\ref{lem:dual}). 
Let $\left\{\left|j\right\rangle\right\}$ be the orthonormal basis of $\mathds{C}^d$ in which $\kappa$ acts as complex conjugation.
For any $j \in \left\{1,\dots,d\right\}$, we must have 
$0 \le {\rm tr}\, T_\nu \left| j \right\rangle\left\langle j \right| \le 1$, since $\left| j \right\rangle\left\langle j \right|$ is a basic effect (a $\kappa$-symmetric projector),
so that
\[
0 \le \left(T_\nu\right)_{jj} \le 1.
\]
Let us assume for some $j,k \in \left\{1,\dots,d\right\}$ and $j\ne k$, 
we have $c_1 := \left\langle k\middle| T_\nu \middle| j\right\rangle \ne 0$. Let
$Q: = \left|j\right\rangle\left\langle j\right| - 2\left|j\right\rangle\left\langle k\right|/c_1$. 
Observe that $Q$ has real entries and
$Q^2 = Q$; therefore, $Q$ is a $\kappa$-symmetric projector. 
However, ${\rm tr}\,T_\nu Q = \left(T_\nu\right)_{jj} - 2 < 0$, which
leads to a contradiction as $\nu\left(Q\right) = {\rm tr}\,T_\nu Q\in[0,1]$,
$\nu$ being a $\kappa$-symmetric state. 
Therefore, we must have 
$\left\langle k\middle| T_\nu \middle| j\right\rangle = 0$ for all pairs $j,k$ with
$j \ne k$. We have concluded that $T_\nu$ is a diagonal matrix. 

Let us next assume that for some $j \ne k$,
$c_2 := \left\langle j\middle| T_\nu \middle| j\right\rangle - \left\langle k\middle| T_\nu \middle| k\right\rangle
\ne 0$. 
Define $\left|\pm\right\rangle := \left(\left|j\right\rangle \pm \left|k\right\rangle\right)/\sqrt{2}$,
and $Q' := \left|+\right\rangle\left\langle +\right| - 3\left|+\right\rangle\left\langle -\right|/c_2$.
Once again, it is easy to verify that $Q'$ is a $\kappa$-symmetric projector.
We have $\left\langle + \middle| T_\nu \middle|+\right\rangle = {\rm tr}\,T_\nu\ket{+}\bra{+}\le 1$ and 
$\left\langle - \middle| T_\nu \middle|+\right\rangle = c_2/2$,
where we have used the fact that $T_\nu$ is diagonal to derive the latter equation.
Therefore,
\[
{\rm tr}\,T_\nu Q' = \left\langle + \middle| T_\nu \middle|+\right\rangle - 
\frac{3}{c_2}\left\langle - \middle| T_\nu \middle|+\right\rangle \le 
1 - \frac{3}{2} <0. 
\]
We reach a contradiction again, and therefore $c_2 = 0$.
We have therefore proved that all diagonal entries of $T_\nu$ must be
identical. Now, from $\nu\left(\mathds{1}\right) = 1$, we get $T_\nu = \mathds{1}/d$,
which leads to $\nu\left(Q\right) = {\rm tr}\,Q/d = {\rm rk}\,Q/d$ as required.
\end{proof}
We now extend this lemma to general PT-symmetries (denoted by the operator $K$) beyond $\kappa$-symmetry.
This leads us to our main result, 
namely that if we replace Hermiticity with $K$-symmetry, 
the system ${\rm A}_{K}$ has only a single valid state.
\begin{thm}
\label{thm:PT}
For every system ${\rm A}_{K}$ of a ${K}$-symmetric theory, 
there exists only one state \(\mu\), given by 
\begin{equation*}
\mu\left(P\right)= \frac{1}{d}\,{\rm rk}\,P
\end{equation*}
for every $K$-symmetric projector $P$, where $d$ is the dimension of the system.
\end{thm}
\begin{proof}
As a first step, let us prove that a system $\mathrm{A}_K$ is equivalent to a system $\mathrm{A}_\kappa$. To this end, define $M := {K}\kappa$. By Theorem 3.1 in Ref.~\cite{HONG1988143}, 
we can express $M$ as $M = S\left(S^*\right)^{-1}$ for some operator $S \in M_d\left(\mathds{C}\right)$, 
where $^*$ denotes complex conjugation of the matrix entries,
so that ${K} = S\kappa S^{-1}$. 
With every ${K}$-symmetric projector $P$, we can associate a projector 
$Q = S^{-1}PS$. 
Now, 
\begin{align*}
Q^* &= \kappa S^{-1}PS \kappa = S^{-1} KP S \kappa \\
&= S^{-1} PK S\kappa =  S^{-1} P S = Q,
\end{align*}
where we have used ${K} = S\kappa S^{-1}$ repeatedly, and also used the fact that $P$ is $K$-symmetric. 
We conclude that $Q$ has real entries,
and it is $\kappa$-symmetric. So in this case $\mathcal{T}:\mathsf{Eff}\left(\mathrm{A}_K\right)\to \mathsf{Eff}\left(\mathrm{A}_\kappa\right)$ is $\mathcal{T}\left(P\right)=S^{-1}PS$, for any $P\in\mathsf{Eff}\left(\mathrm{A}_K\right)$. Notice that $\mathcal{T}$ is a linear bijection (its inverse is $\mathcal{T}^{-1}\left(Q\right)=SQS^{-1}$, for $Q\in\mathsf{Eff}\left(\mathrm{A}_{\kappa}\right)$) and $\mathcal{T}\left(\mathds{1}\right)=S^{-1}\mathds{1}S=\mathds{1}$. Therefore, $\mathrm{A}_K$ is equivalent to $\mathrm{A}_\kappa$.

Then we know that there is a (linear) bijection between the sets of states of $\mathrm{A}_\kappa$ and $\mathrm{A}_K$. Hence, $\mathrm{A}_K$ will have one state $\mu$ too. To determine it, we make use of the dual map of $\mathcal{T}$, as per Lemma~\ref{lem:duality}. We then have, for every $P\in\mathsf{Eff}\left(\mathrm{A}_{K}\right)$,
\[
\mu\left(P\right) = \nu\left(S^{-1} P S\right) = {\rm rk}\left(S^{-1} P S\right)/d = {\rm rk}\,P/d,
\]
where $\nu$ is the state determined in Lemma~\ref{lem:realprojectors}.

\end{proof}

This theorem shows that a purely PT-symmetric theory is trivial, therefore PT-symmetry alone does not extend quantum theory in any meaningful way. Finally, we can represent the unique state from Theorem~\ref{thm:PT} by a multiple of the identity matrix,
thanks to Lemma~\ref{lem:linear algebra}.
\begin{coro}
The unique state of a $d$-dimensional ${K}$-symmetric system ${\rm A}_{K}$
can be represented by $\frac{1}{d}\mathds{1}$. 
\end{coro}
We conclude this section with a remark that our analysis can be extended to more general 
antilinear involutions $K$ as the unbroken symmetry of observables. In other words, we go beyond the case of PT-symmetry. To see this, note that Theorem~\ref{thm:PT}, which constitutes the core of the results presented in this 
section, holds for \emph{any} antilinear involution.

\section{A GPT with quasi-Hermitian effects}\label{sec:quasi-hermitian}
After the failure of PT-symmetry alone to extend quantum theory, we begin our journey to explore other possible ways, related to PT-symmetry, to extend quantum theory.
Specifically, in this section we show that if the allowed observables on a certain Hilbert space are quasi-Hermitian (also known as $\eta$-Hermitian,
cf.\ Definition~\ref{def:qH})
then such a system is equivalent to a standard quantum system.
We show that the states in this theory 
are also quasi-Hermitian ($\eta$-Hermitian) with respect to the same~$\eta$.
In order to emphasize the $\eta$-dependence of the quasi-Hermiticity constraint,
we refer to the operators satisfying Definitions~\ref{def:qH} and~\ref{def:qU}
as $\eta$-Hermitian and $\eta$-unitary, respectively. 
 Note that the equivalence of 
quasi-Hermitian quantum systems with standard quantum systems was already known \cite{KAS22}.
Nevertheless, here we rederive this result in the broader framework of general probabilistic theories, which subsumes the known result.
It is worth emphasizing that our analysis does not make any a-priori assumption that pure states of 
quasi-Hermitian quantum theory form a Hilbert space.
 
Given that this new theory only admits $\eta$-Hermitian observables,
we characterize the set of effects and states allowed for the system.
In order to do so, we need the following definitions.
\begin{defn}
\label{def:etale}
An $\eta$-Hermitian operator $E$ is \emph{$\eta$-positive semidefinite}, denoted $E \ge_{\eta} \mathbf{0}$, if
\begin{equation*}
    \left\langle\psi\middle| E\middle|\psi\right\rangle_\eta \ge 0 \quad \forall \left|\psi\right\rangle \in \mathds{C}^d.
\end{equation*}
\end{defn}
Note that $\left\langle\psi\middle| E\middle|\psi\right\rangle_\eta = \left\langle\psi\middle| \eta E\middle|\psi\right\rangle$
by the definition of the $\eta$-inner product in Eq.~\eqref{eq:etaIP}, so $E$ is $\eta$-positive semidefinite if and only if $\eta E$ is positive semidefinite.
\begin{defn}
An \emph{$\eta$-density matrix} is a $\eta$-positive semidefinite matrix of unit trace. 
\end{defn}

For any $\eta$-Hermitian observable $O$, 
it is not hard to see, with the help of the modified inner product
in Eq.~\eqref{eq:etaIP}, that $O$ has a spectral decomposition in terms of 
rank-1 projectors that are also $\eta$-Hermitian.
Therefore, basic effects in the new theory can be taken to be  all
rank-1, $\eta$-Hermitian projectors.
More generally, the set of all allowed effects of this system $\mathrm{A}_{\eta}$  
is given by
\begin{equation*}
    \mathsf{Eff}\left(\mathrm{A}_{\eta}\right) = \left\{E: E^\dagger = \eta E \eta^{-1}, \mathbf{0} \le_{\eta} E \le_{\eta} \mathds{1}\right\}.
\end{equation*}
In this setting, measurements are all the collections of $\eta$-Hermitian effects 
that sum to $\mathds{1}$. Basic measurements are those comprised of rank-1 $\eta$-Hermitian projectors.
As $\left\{\mathds{1}\right\}$ is also a measurement, any valid state $\nu$ of system $\mathrm{A}_{\eta}$ must obey the property $\nu\left(\mathds{1}\right)=1$ (again, $\mathds{1}$ is the unit effect).

We now prove that every $\eta$-Hermitian quantum system $\mathrm{A}_{\eta}$
is equivalent to a standard, 
i.e.\ Hermitian quantum system $\mathrm{A}_{\mathds{1}}$ ($\eta=\mathds{1}$).

\begin{lem}\label{lem:bijection quasi-Hermitian}
The systems $\mathrm{A}_{\eta}$ and $\mathrm{A}_{\mathds{1}}$ are equivalent.
\end{lem}
\begin{proof}
Let us consider the map $\mathcal{T}: \mathsf{Eff}\left(\mathrm{A}_{\eta}\right) \to \mathsf{Eff}\left(\mathrm{A}_{\mathds{1}}\right)$, whereby $E \mapsto \eta^{1/2}E\eta^{-1/2}$, for every $E\in\mathsf{Eff}\left(\mathrm{A}_{\eta}\right)$. Let us check that  $\eta^{1/2}E\eta^{-1/2}$ is Hermitian, and that $\mathbf{0}\leq\eta^{1/2}E\eta^{-1/2} \le \mathds{1}$. 
The Hermiticity of $\eta^{1/2}E\eta^{-1/2}$ can be proven by
\begin{align*}
    \left(\eta^{1/2}E\eta^{-1/2}\right)^\dagger &= \eta^{-1/2}E^\dagger\eta^{1/2} \\ &= \eta^{-1/2}\eta E \eta^{-1}\eta^{1/2} \\&= \eta^{1/2}E\eta^{-1/2},
\end{align*}
where we have used the fact that $E$ is $\eta$-Hermitian.
The property $\eta^{1/2}E\eta^{-1/2} \geq \mathbf{0}$ follows from
\begin{align*}
\left\langle\phi\middle|\eta^{1/2}E\eta^{-1/2} \middle|\phi\right\rangle &= 
   \left\langle\psi\middle|\eta^{1/2} \left(\eta^{1/2}E \eta^{-1/2}\right) \eta^{1/2}\middle|\psi\right\rangle \nonumber\\
    &=  \left\langle\psi\middle|\eta E\middle|\psi\right\rangle \geq 0  \quad \forall \ket{\phi} \in \mathds{C}^d,
\end{align*}
where we have used the substitution $\ket{\psi} := \eta^{-1/2}\ket{\phi}$ in the first equality 
and Definition~\ref{def:etale} in the last step.
The property $\eta^{1/2}E\eta^{-1/2} \le \mathds{1}$ is proven in a similar way:
\begin{align*}
&\left\langle\phi\middle|\left(\eta^{1/2}E\eta^{-1/2} - \mathds{1}\right)\middle|\phi\right\rangle \\&= 
    \left\langle\psi\middle|\eta^{1/2} \left[\eta^{1/2}\left(E - \mathds{1}\right)\eta^{-1/2}\right] \eta^{1/2}\middle|\psi\right\rangle \nonumber\\
    &= \left\langle\psi\middle|\eta\left(E - \mathds{1}\right)\middle|\psi\right\rangle \le 0  \quad \forall \ket{\phi} \in \mathds{C}^d.
\end{align*}
$\mathcal{T}$ is a linear bijection: the inverse is $\mathcal{T}^{-1}\left(F\right)= \eta^{-1/2}F\eta^{1/2}$, for all $F\in \mathsf{Eff}\left(\mathrm{A}_{\mathds{1}}\right)$. Finally, $\mathcal{T}\left(\mathds{1}\right)=\eta^{1/2}\mathds{1}\eta^{-1/2}=\mathds{1}$, which concludes the proof.
\end{proof}
This result is already sufficient to conclude that quasi-Hermiticity does not provide any meaningful extension of quantum theory, as systems are equivalent. Lemma~\ref{lem:bijection quasi-Hermitian} implies that there is a linear bijection between the corresponding set of states, which we can exploit to derive the states of a quasi-Hermitian system.
\begin{prop}
\label{lem:QHdual}
For every state  
$\mu \in  \mathsf{St}\left(\mathrm{A}_{\eta}\right)$, 
there exists a unique 
$\eta$-density matrix $\rho_\mu$ satisfying
\[
\mu\left(E\right) = \mathrm{tr}\,\rho_\mu E \qquad \forall E \in \mathsf{Eff}\left(\mathrm{A}_{\eta}\right).
\]
\end{prop}
\begin{proof}
By Lemma~\ref{lem:duality}, we know that there is a linear bijection $\mathcal{T}':\mathsf{St}\left(\mathrm{A}_{\mathds{1}}\right)\rightarrow\mathsf{St}\left(\mathrm{A}_{\eta}\right)$ constructed as the dual of the map $\mathcal{T}$ introduced in Lemma~\ref{lem:bijection quasi-Hermitian}. Consider $\nu\in \mathsf{St}\left(\mathrm{A}_{\mathds{1}}\right)$, which is such that $\nu\left(F\right)=\mathrm{tr}\:\sigma_{\nu}F$, where $\sigma_\nu$ is its associated density matrix, and $F\in\mathsf{Eff}\left(\mathrm{A}_{\mathds{1}}\right)$. Then, $\mu \in  \mathsf{St}\left(\mathrm{A}_{\eta}\right)$ is constructed as
\begin{align}\label{eq: T quasi-Hermitian}
\mu\left(E\right)&=\nu\left(\mathcal{T}\left(E\right)\right)=\nu\left(\eta^{1/2}E\eta^{-1/2}\right)\nonumber\\&=\mathrm{tr}\:\sigma_{\nu}\eta^{1/2}E\eta^{-1/2}=\mathrm{tr}\:\eta^{-1/2}\sigma_{\nu}\eta^{1/2}E,
\end{align}
where $E\in\mathsf{Eff}\left(\mathrm{A}_{\eta}\right).$

Therefore, one choice for the matrix representation of $\mu$ is $\rho_\mu := \eta^{-1/2}\sigma_\nu\eta^{1/2}$, which can be easily verified
to be $\eta$-Hermitian. Furthermore, $\rho_\mu \ge_{\eta} 0$ which follows from 
\begin{equation*}
    \left\langle\psi\middle|\eta \rho_\mu\middle|\psi\right\rangle = \left\langle\psi\middle|\eta\left(\eta^{-1/2}\sigma_\nu\eta^{1/2}\right)\middle|\psi\right\rangle
    =: \left\langle\phi\middle|\sigma_\nu\middle|\phi\right\rangle \ge 0,
\end{equation*}
where we have set $\left|\phi\right\rangle:=\eta^{1/2}\left|\psi\right\rangle$.
Finally, $\mu\left(\mathds{1}\right) = 1$ implies $\mathrm{tr}\,\rho_\mu=1$. Therefore $\rho_\mu$ is an $\eta$-density matrix. The uniqueness of $\rho_{\mu}$ follows from Lemma~\ref{lem:linear algebra}.
\end{proof}
With this proposition we concluded that a system with $\eta$-Hermitian observables leads to states that are
  represented by $\eta$-density matrices. 
 

\section{GPT with a combination of PT-symmetric and quasi-Hermitian constraints on  effects}\label{sec:PT+quasi}
In Section~\ref{sec:GPT-PT} we proved that the constraint of PT-symmetry alone
on observables gives rise to a trivial theory.
Therefore,
we now consider a system that is quasi-Hermitian for some $\eta$,
and then we impose the constraint of PT-symmetry on observables.
We model the constraint of PT-symmetric invariance on an $\eta$-Hermitian observable by introducing
an $\eta$-antiunitary operator.
\begin{defn}\label{def:etaantiunitary}
An antilinear operator $K_\eta$ is $\eta$-antiunitary if
\begin{equation*}
    \left\langle K_\eta\psi\middle|K_\eta\phi\right\rangle_\eta = \left\langle \psi\middle|\phi\right\rangle_\eta^* \quad \forall \left|\psi\right\rangle,\left|\phi\right\rangle \in \mathds{C}^d,
\end{equation*}
where $\langle\bullet |\bullet\rangle_\eta$ is the $\eta$-inner product defined in Eq.~\eqref{eq:etaIP}.
\end{defn}
\noindent In this section, we denote by $K_\eta$ 
any valid $\eta$-antiunitary operator that serves as a PT-symmetry in the $\eta$-inner product.
Note that if $\eta = \mathds{1}$, then the PT operator $K_\eta$ is antiunitary,
which is consistent with the literature,
as discussed in Section~\ref{sec:background}.

The main finding of this section is that a system with $\eta$-Hermitian, $K_\eta$-symmetric observables
is isomorphic to a real quantum system. 
To prove this result, we first focus on the special case 
where the allowed observables are Hermitian ($\eta=\mathds{1}$)
as well as $\kappa$-symmetric, $\kappa$ being complex conjugation in the canonical basis as in Section~\ref{sec:GPT-PT},
and show that we arrive at a real quantum system \cite{Wootters-real,Hardy-real,Baez,Hickey,Barnum2020composites,Scandolo-real1,Scandolo-real2,Acin-complex}.
After that, we extend the analysis to observables being Hermitian
as well as $K$-symmetric, where $K$ is any valid PT-symmetry (cf. Section~\ref{subsec:PT-symmetry}).
Finally we consider the case where  observables are $\eta$-Hermitian, for $\eta\neq\mathds{1}$ and $K_\eta$-symmetric,
and show that the resulting system is equivalent to a real quantum system. 

We first discuss how the constraints of $\eta$-Hermiticity and  $K_\eta$-symmetry on  observables translate into
constraints on the allowed effects on the system.
Any observable $O$ that is $\eta$-Hermitian as well as $K_\eta$-symmetric 
has a spectral decomposition in terms of
rank-1 projectors that are also $\eta$-Hermitian and $K_\eta$-symmetric (these are basic effects).
This observation follows from restricting Proposition~\ref{prop:PTspectral} to an $\eta$-Hermitian observable $O$
with unbroken $K_\eta$-symmetry.
Consequently, the set of all allowed effects of the system are $K_\eta$-symmetric 
and $\eta$-positive semidefinite, satisfying in addition $E \le_{\eta} \mathds{1}$.
For the special case of Hermitian ($\eta=\mathds{1})$, $\kappa$-symmetric observables we first focus on,
this characterization implies that the set of effects of the system A$_{\kappa,\mathds{1}}$
are given by
\begin{equation*}
    \mathsf{Eff}\left(\mathrm{A}_{\kappa,\mathds{1}}\right) = \left\{E: E^\dagger = E, E^* = E, \mathbf{0} \le E \le \mathds{1}\right\}.
\end{equation*}
Now we show that for the set of effects $\mathsf{Eff}\left(\mathrm{A}_{\kappa,\mathds{1}}\right)$,
the allowed set of states are density matrices with real entries in the canonical basis.
\begin{lem}\label{lem:real density matrix}
Each state $\nu \in \mathsf{St}\left(\mathrm{A}_{\kappa,\mathds{1}}\right)$ can be represented by a $\kappa$-symmetric, i.e., real density matrix.
\end{lem}
\begin{proof}
We need to show that every $\nu \in \mathsf{St}\left(\mathrm{A}_{\kappa,\mathds{1}}\right)$ is 
represented by a $\sigma_\nu \in M_d\left(\mathds{R}\right)$ with $\sigma_\nu \ge 0$ and $\mathrm{tr}\,\sigma_\nu = 1$. 
The existence of a $\sigma_\nu \in M_d\left(\mathds{R}\right)$ is
a direct consequence of Lemma~\ref{lem:dual}. The normalization condition $\mathrm{tr}\,\sigma_\nu = 1$ follows
from $\mathrm{tr}\,\sigma_\nu = \nu\left(\mathds{1}\right) = 1$. Finally, 
\begin{equation*}
    \left\langle\psi\middle|\sigma_\nu\middle|\psi\right\rangle =\mathrm{tr}\,\sigma_\nu \ket{\psi}\bra{\psi} \in \left[0,1\right] \quad \forall \ket{\psi} \in \mathds{R}^d,
\end{equation*}
because $\ket{\psi}\bra{\psi} \in \mathsf{Eff}\left(\mathrm{A}_{\kappa,\mathds{1}}\right)$.
\end{proof}

We now move to the case in which  observables are $K$-symmetric 
and Hermitian.
We denote this system by $\mathrm{A}_{K,\mathds{1}}$  and the allowed set of effects is given by
\begin{equation*}
    \mathsf{Eff}\left(\mathrm{A}_{K,\mathds{1}}\right) = \left\{E: E^\dagger = E, KE = EK, \mathbf{0} \le E \le \mathds{1}\right\}.
\end{equation*}
Now we show that  system $\mathrm{A}_{K,\mathds{1}}$
is equivalent to a real quantum system.
\begin{prop}\label{prop:isomorphism real effects}
The systems $\mathrm{A}_{\kappa,\mathds{1}}$ and $\mathrm{A}_{\kappa,\mathds{1}}$ are equivalent.
\end{prop}
\begin{proof}
We prove this proposition by constructing a bijection $\mathcal{T}:\mathsf{Eff}\left(\mathrm{A}_{K,\mathds{1}}\right)\to
\mathsf{Eff}\left(\mathrm{A}_{\kappa,\mathds{1}}\right)$. 
To begin, recall $\kappa^2=\mathds{1}$. Then $K = \left(K\kappa\right)\kappa=:U\kappa$.
$U$ is linear (it is the composition of two antilinear operators) and unitary, as it is the composition of two antiunitary operators (cf.\ Lemma~\ref{lem:productantiunitary}). 
The fact that  $K^2 = \mathds{1}$ implies that 
$U\kappa U\kappa=UU^*=\mathds{1}$.
This implies that $U^*=U^{\dagger}$, from which $U = U^{\rm T}$. 
Then by Autonne-Takagi factorization~\cite{horn1994},
$U = VV^{\rm T}$ for some unitary matrix $V$, so that
\begin{equation}
\label{eq:PTdecomposition}
    K =U\kappa=  VV^{\rm T}\kappa= V\kappa \left(\kappa V^{\rm T}\kappa\right) = V\kappa V^\dagger.
\end{equation}
We now show that
\begin{align}\label{eq:T}
    \mathcal{T}:&\mathsf{Eff}\left(\mathrm{A}_{K,\mathds{1}}\right) \to \mathsf{Eff}\left(\mathrm{A}_{\kappa,\mathds{1}}\right) \nonumber\\ 
    &E \mapsto V^\dagger EV
\end{align}
is the required bijection.
Observe that this definition is well posed:
\begin{align*}
    \kappa \mathcal{T}\left(E\right) \kappa &= \left(\kappa V^\dagger\right) E \left(V \kappa\right)  \nonumber\\
    &= \left(V^\dagger K\right)E\left(K V\right) \nonumber\\
    &= \left(V^\dagger K\right)K^{-1}EK\left(KV\right) \nonumber\\ 
    &= V^\dagger E V\\
    &= \mathcal{T}\left(E\right),
\end{align*}
where the second equality follows from Eq.~\eqref{eq:PTdecomposition}, and the third and fourth ones follow by using $KE = EK$ and 
$K^2 = \mathds{1}$    respectively.
Therefore, $\mathcal{T}\left(E\right)$ is a matrix with real entries. 
Clearly $V^\dagger E V$ is Hermitian, and
as $\mathcal{T}$ is a similarity transformation, 
$\mathcal{T}\left(E\right)$ and $E$ have the same spectrum, so $\mathbf{0}\leq \mathcal{T}\left(E\right) \le \mathds{1}$. This shows that $\mathcal{T}\left(E\right)\in\mathsf{Eff}\left(\mathrm{A}_{\kappa,\mathds{1}}\right)$. Further,  $\mathcal{T}$ is a linear bijection, with inverse $F\mapsto VFV^\dagger$, $F\in\mathsf{Eff}\left(\mathrm{A}_{\kappa,\mathds{1}}\right)$. Finally, $\mathcal{T}\left(\mathds{1}\right)=V^\dagger \mathds{1}V=\mathds{1}$. This completes the proof of the equivalence.
\end{proof}

As a corollary, we have a linear bijection between the corresponding set of states, which we will now use to characterize the states of the system $\mathrm{A}_{K,\mathds{1}}$ in matrix form.

\begin{thm}\label{thm:K-symmetric matrices}
Each state $\mu \in \mathsf{St}\left(\mathrm{A}_{K,\mathds{1}}\right)$ can be represented by a $K$-symmetric density matrix.
\end{thm}
\begin{proof}
We employ the usual dual construction of Lemma~\ref{lem:duality}: we consider $\mathcal{T}':\mathsf{St}\left(\mathrm{A}_{\kappa,\mathds{1}}\right)\rightarrow\mathsf{St}\left(\mathrm{A}_{K,\mathds{1}}\right)$. Take $\nu\in \mathsf{St}\left(\mathrm{A}_{\kappa,\mathds{1}}\right)$, which is such that $\nu\left(F\right)=\mathrm{tr}\:\sigma_{\nu}F$, where $\sigma_\nu$ is its associated real density matrix (see Lemma~\ref{lem:real density matrix}), and $F\in\mathsf{Eff}\left(\mathrm{A}_{\kappa,\mathds{1}}\right)$. Then, $\mu \in  \mathsf{St}\left(\mathrm{A}_{K,\mathds{1}}\right)$ is constructed as
\begin{align*}
\mu\left(E\right)&=\nu\left(\mathcal{T}\left(E\right)\right)=\nu\left( V^\dagger EV\right)\\&=\mathrm{tr}\:\sigma_{\nu} V^\dagger EV=\mathrm{tr}\:V\sigma_{\nu}V^\dagger E,
\end{align*}
where $E\in\mathsf{Eff}\left(\mathrm{A}_{K,\mathds{1}}\right)$.
Therefore, one choice for $\rho_\mu$ is $\rho_\mu := V \sigma_\nu V^\dagger$. Clearly $\rho_\mu \ge 0$ and $\mathrm{tr}\,\rho_\mu = 1$. 
We now show that $\rho_\mu$ is
$K$-symmetric. This is revealed by 
\begin{align*}
    K\rho_\mu K &= \left(V\kappa V^\dagger\right) V \sigma_\nu V^\dagger \left(V\kappa V^\dagger\right) \nonumber \\
    &= V\kappa  \sigma_\nu \kappa V^\dagger \nonumber \\
    &= V\sigma_\nu V^\dagger \\
    &= \rho_\mu,
\end{align*}
where we have used Eq.~\eqref{eq:PTdecomposition}, and $\kappa \sigma_\nu \kappa = \sigma_\nu$ in the third equality.
\end{proof}

We finally come to the most general case, in which effects are $K_\eta$-symmetric and $\eta$-Hermitian.
We assume that $K_\eta$ is an $\eta$-antiunitary operator (cf.\ Definition~\ref{def:etaantiunitary}),
taking the role of PT-symmetry in the $\eta$-inner product. In this case, we have
\begin{align*}
 &\mathsf{Eff}\left(\mathrm{A}_{K_\eta,\eta}\right)\\&= \left\{E: E^\dagger = \eta E \eta^{-1}, K_\eta E = EK_\eta, \mathbf{0} \le_\eta E \le_\eta \mathds{1}\right\}.
\end{align*}
We next prove a few lemmas required for proving the main result.
The first of these lemmas constructs an $\eta$-equivalent of complex conjugation. 
\begin{lem}
The operator $\kappa_\eta = \eta^{-1/2}\kappa\eta^{1/2}$ is $\eta$-antiunitary. 
Moreover, $\kappa_\eta^2=\mathds{1}$.
\end{lem}
\begin{proof}First of all, note that $\kappa_\eta$ is antilinear due to the presence of $\kappa$. The proof follows from 
\begin{align*}
     \left\langle\kappa_\eta\psi\middle|\kappa_\eta\phi\right\rangle_\eta &= \left\langle\kappa_\eta\psi\middle|\eta \eta^{-1/2}\left(\kappa\eta^{1/2}\kappa\right)\left(\kappa\phi\right)\right\rangle \\
    &= \left\langle\kappa_\eta\psi\middle|\eta^{1/2}(\eta^{1/2})^*\phi^*\right\rangle,
\end{align*}
where we have used the definitions of $\eta$-inner product and of $\kappa_\eta$,
and the properties of $\kappa$. Now
\begin{align*}
     \left\langle\kappa_\eta\psi\middle|\kappa_\eta\phi\right\rangle_\eta &=\left\langle\phi^*\middle|\left(\eta^{1/2}\right)^*\eta^{1/2}\middle|\kappa_\eta\psi\right\rangle^* \nonumber\\
    &= \left\langle\phi^*\middle|\left(\eta^{1/2}\right)^*\eta^{1/2}\eta^{-1/2}\kappa\eta^{1/2}\middle|\psi\right\rangle^*\\ \nonumber
    &= \left\langle\phi^*\middle|\left(\eta^{1/2}\right)^*\left(\eta^{1/2}\right)^*\middle|\psi^*\right\rangle^*
    \end{align*}
    where we have used the properties of $\kappa$ and the definition of $\kappa_\eta$ again. The properties of complex conjugation yield
\begin{equation*}  
      \left\langle\kappa_\eta\psi\middle|\kappa_\eta\phi\right\rangle_\eta = 
    \left\langle\phi\middle|\eta\middle|\psi\right\rangle
   = \left\langle\psi\middle|\phi\right\rangle_\eta^*.
\end{equation*}
The second property of $\kappa_\eta$, namely $\kappa_\eta^2 = \mathds{1}$, follows from
\begin{equation*}
    \kappa_\eta^2=\eta^{-1/2}\kappa\eta^{1/2}\eta^{-1/2}\kappa\eta^{1/2}=\eta^{-1/2}\kappa^2\eta^{1/2}=\mathds{1}.
\end{equation*}
\end{proof}
The next two lemmas express some useful properties of $\eta$-antiunitary operators.
\begin{lem}
\label{lem:productantiunitary}
The product of two $\eta$-antiunitary operators is $\eta$-unitary.
\end{lem}
\begin{proof}
Let $K_\eta^{\left(1\right)},K_\eta^{\left(2\right)} \in M_d\left(\mathds{C}\right)$ be two $\eta$-antiunitary operators. Invoking the definition of $\eta$-antiunitarity twice,
we get
\begin{align*}
    \left\langle K_\eta^{(1)}K_\eta^{(2)}\psi\middle| K_\eta^{(1)}K_\eta^{(2)}\phi\right\rangle_\eta &= \left\langle K_\eta^{(2)}\psi\middle| K_\eta^{(2)}\phi\right\rangle_\eta^* \\&= \left\langle\psi\middle|\phi\right\rangle_\eta \quad
    \forall \ket{\psi},\ket{\phi} \in \mathds{C}^d,
\end{align*}
which proves that $K_\eta^{\left(1\right)}K_\eta^{\left(2\right)}$ is an $\eta$-unitary operator.
\end{proof}
\begin{lem}
\label{lem:etaantiunitary}
Any $\eta$-antiunitary operator $K_\eta$ can be expressed as $K_\eta=U_\eta\kappa_\eta$, where $U_\eta$ is $\eta$-unitary.
\end{lem}
\begin{proof}
By the properties of $\kappa_\eta$, we have $K_\eta = \left(K_\eta\kappa_\eta\right)\kappa_\eta$, and  
$U_\eta := K_\eta\kappa_\eta$ is $\eta$-unitary by Lemma~\ref{lem:productantiunitary}. 
\end{proof}
The next lemma links $\eta$-antiunitary operators with standard antiunitary operators.
\begin{lem}
\label{lem:Santiunitary}
If $K_\eta$ is an $\eta$-antiunitary operator, then $K := \eta^{1/2}K_\eta\eta^{-1/2}$ is a standard antiunitary operator.
\end{lem}
\begin{proof}
We know that $K_\eta\kappa_\eta$ is $\eta$-unitary by Lemma~\ref{lem:etaantiunitary}, 
and therefore $\left(K_\eta\kappa_\eta\right)^\dagger\eta \left(K_\eta\kappa_\eta\right) \eta^{-1} = \mathds{1}$ by Definition~\ref{def:qU}. This implies that we also have
\begin{equation}
\label{eta-antiunitaryAeta}
 \left(K_\eta\kappa_\eta\right) \eta^{-1}\left(K_\eta\kappa_\eta\right)^\dagger\eta  = \mathds{1}.
\end{equation}
The left-hand side of this equation can be simplified to
\begin{align*}
    &\left(K_\eta\kappa_\eta\right) \eta^{-1}\left(K_\eta\kappa_\eta\right)^\dagger\eta \\
    &=  K_\eta \left(\eta^{-1/2}\kappa \eta^{1/2}\right) \eta^{-1} \left[K_\eta \left(\eta^{-1/2}\kappa \eta^{1/2}\right)\right]^\dagger\eta \nonumber\\
    &=  K_\eta \kappa \left(\eta^{-1/2}\right)^*\eta^{1/2} \eta^{-1}\left[K_\eta \kappa\left(\eta^{-1/2}\right)^* \eta^{1/2}\right]^\dagger\eta  \nonumber\\
    &= K_\eta \kappa \left(\eta^{-1/2}\right)^*\eta^{1/2} \eta^{-1}\eta^{1/2}\left(\eta^{-1/2}\right)^*\left(K_\eta \kappa \right)^\dagger\eta  \nonumber\\
    &=  K_\eta\kappa \left(\eta^{-1}\right)^*\left(K_\eta\kappa\right)^\dagger \eta,
\end{align*}
which allows us to recast Eq.~\eqref{eta-antiunitaryAeta} as
\begin{equation}
\label{eq:step}
    K_\eta\kappa \left(\eta^{-1}\right)^*\left(K_\eta\kappa\right)^\dagger \eta = \mathds{1}.
\end{equation}
Now we are ready to show that $K\kappa$ is a unitary, which means that $K$ is antiunitary.
\begin{align*}
     K\kappa \left(K\kappa\right)^\dagger &= \eta^{1/2}K_\eta\eta^{-1/2}\kappa\left(\eta^{1/2}K_\eta\eta^{-1/2}\kappa\right)^\dagger \nonumber\\
    &= \eta^{1/2}K_\eta\kappa\left(\eta^{-1/2}\right)^*\left[\eta^{1/2}K_\eta\kappa\left(\eta^{-1/2}\right)^*\right]^\dagger \nonumber\\
    &= \eta^{1/2}K_\eta\kappa\left(\eta^{-1/2}\right)^*\left(\eta^{-1/2}\right)^* \left(K_\eta\kappa\right)^\dagger \eta^{1/2} \nonumber\\
    &= \eta^{1/2} \left[K_\eta \kappa  \left(\eta^{-1}\right)^*\left(K_\eta\kappa\right)^\dagger \eta \right] \eta^{-1/2} \nonumber \\
    &= \eta^{1/2}\mathds{1}\eta^{-1/2} \nonumber\\
    &= \mathds{1},
\end{align*}
where we have used Eq.~\eqref{eq:step} in the second last equality.
\end{proof}

We now prove the key proposition, which establishes the equivalence of the systems
${\rm A}_{K_\eta,\eta}$ and ${\rm A}_{K,\mathds{1}}$, where $K_\eta$ is an $\eta$-antiunitary operator and $K$, defined in Lemma~\ref{lem:Santiunitary},
represents PT-symmetry. 
\begin{prop}
The systems $\mathrm{A}_{K_\eta,\eta}$ and $\mathrm{A}_{K,\mathds{1}}$,  where $K= \eta^{1/2}K_\eta\eta^{-1/2}$, are equivalent.
\end{prop}
\begin{proof}
We prove the statement by constructing  a bijection explicitly. 
\begin{align}\label{eq:T again}
    \mathcal{T}:&\mathsf{Eff}\left(\mathrm{A}_{K_\eta,\eta}\right) \to \mathsf{Eff}\left(\mathrm{A}_{K,\mathds{1}}\right)\nonumber \\
    &E \mapsto \eta^{1/2}E\eta^{-1/2}.
\end{align}
In Lemma~\ref{lem:bijection quasi-Hermitian} we proved that $\mathcal{T}$ is a linear bijection, $\mathcal{T}\left(\mathds{1}\right)=\mathds{1}$, and that $\mathcal{T}\left(E\right)$ is Hermitian and such that $\mathbf{0}\leq \mathcal{T}\left(E\right) \leq \mathds{1}$, for all $\eta$-Hermitian effects $E$, hence also for all $\eta$-Hermitian effects $E$ that are also $K_\eta$-symmetric. We are only left to show that $\mathcal{T}\left(E\right)$ is $K$-symmetric.
\begin{align*}
   & K\mathcal{T}\left(E\right) K \\&= K\eta^{1/2}E\eta^{-1/2}K \nonumber \\
    &= K\eta^{1/2}K_\eta EK_\eta \eta^{-1/2}K \nonumber\\
    &= K\eta^{1/2}\left(\eta^{-1/2}K\eta^{1/2}\right) E\left(\eta^{-1/2}K\eta^{1/2}\right)\eta^{-1/2}K \nonumber \\
    &=\mathcal{T}\left(E\right).
\end{align*}
Here, in the second equality, we have used the fact that $E$ is $K_\eta$-symmetric, and in the third equality the definition of $K$.
\end{proof}
Thanks to this proposition, we have proved the main result in this section: 
a system with effects that are quasi-Hermitian and PT-symmetric 
is equivalent to a real quantum system.
We conclude our analysis of this new theory by
characterizing the matrix representation of states in $\mathsf{St}\left(\mathrm{A}_{K_\eta,\eta}\right)$
through the following theorem.
\begin{thm}
Each state $\mu \in \mathsf{St}\left(\mathrm{A}_{K_\eta,\eta}\right)$ can be represented by a $K_\eta$-symmetric $\eta$-density matrix.
\end{thm}
\begin{proof}
As usual, we establish a linear bijection $\mathcal{T}': \mathsf{St}\left(\mathrm{A}_{K,\mathds{1}}\right) \to \mathsf{St}\left(\mathrm{A}_{K_\eta,\eta}\right) $ as per Lemma~\ref{lem:duality}. Take $\nu\in\mathsf{St}\left(\mathrm{A}_{K,\mathds{1}}\right)$. By Theorem~\ref{thm:K-symmetric matrices}, $\nu\left(F\right)=\mathrm{tr}\:\sigma_\nu F$, for $F\in\mathsf{Eff}\left(\mathrm{A}_{K,\mathds{1}}\right)$, where $\sigma_\nu$ is a $K$-symmetric density matrix. Then, for $E\in\mathsf{Eff}\left(\mathrm{A}_{K_\eta,\eta}\right)$, the construction, which is identical to Eq.~\eqref{eq: T quasi-Hermitian}, yields
\[
\mu\left(E\right)=\mathrm{tr}\:\eta^{-1/2}\sigma_{\nu}\eta^{1/2}E.
\]
Therefore, one choice for $\rho_\mu$ is again $\rho_\mu := \eta^{-1/2} \sigma_\nu \eta^{1/2}$. We already know that $\rho_\mu$ is an $\eta$-density matrix from the proof of Proposition~\ref{lem:QHdual}. We only need to show that $\rho_\mu$ is
$K_\eta$-symmetric.
\begin{align*}
    K_\eta \rho_\mu K_\eta &= \eta^{-1/2} K \eta^{1/2}  \left(\eta^{-1/2} \sigma_\nu \eta^{1/2}\right)\eta^{-1/2} K \eta^{1/2} \nonumber\\
    &= \eta^{-1/2} K \sigma_\nu  K \eta^{1/2} \nonumber\\
    &= \eta^{-1/2} \sigma_\nu \eta^{1/2} \nonumber\\
    &= \rho_\mu,
\end{align*}
where we have used the fact that $\sigma_\nu$ is $K$-symmetric.
\end{proof}
We have therefore shown that a GPT with $K_\eta$-symmetric and $\eta$-Hermitian effects  is non-trivial, unlike the theory introduced in Section~\ref{sec:GPT-PT}, and each system is equivalent to a real quantum system. This result holds for any choice of $K_\eta$-symmetry and
the metric operator $\eta$, as long as $K_\eta$ is an $\eta$-antiunitary operator.

\section{Conclusions}\label{sec:conclusions}
In this article, we conclusively answered the question of whether 
a consistent physical theory with PT-symmetric observables could extend standard quantum theory. Indeed, the development of PT-symmetric quantum theory was motivated by the conjecture that replacing the 
ad-hoc condition of Hermiticity of observables with the
physically meaningful constraint of PT-symmetry could lead to a non-trivial extension of standard quantum theory.
However, no such consistent extension of standard quantum mechanics based on PT-symmetric observables
has been formulated to date.

Two approaches for formulating a consistent PT-symmetric quantum mechanics, 
which could potentially result in an extension of standard quantum theory, have
been attempted in the literature.
The first approach leverages quasi-Hermiticity of unbroken PT-symmetric observables.
The quasi-Hermitian approach does not replace the Hermiticity constraint with 
PT-symmetry, but rather imposes Hermiticity on observables with respect to a different inner product.
If, on the one hand, this approach can provide a self-consistent theory, on the other hand, it is equivalent to standard quantum mechanics, and does not 
offer any extension. Another approach to a consistent formulation of PT-symmetric quantum theory
is based on Krein spaces. In contrast to quasi-Hermitian quantum theory,
whether the theories developed within this approach are self-consistent is still an open question.

In this article, we proposed an approach based instead on the framework of general probabilistic theories (GPTs). This framework is applicable 
to any theory that is probabilistic, and is commonly used for studying quantum mechanics 
and other physical theories. We showed that if PT-symmetry is the only constraint on the set of
observables, then the resulting theory has only a single, trivial state.
In a nutshell, the reason behind the set of states being extremely restricted is that
PT-symmetry is a weak constraint on the set of effects, and consequently the set of allowed effects 
is rather large. The dual to the set of effects, namely the set of states, is therefore rather small. 
In fact, the set of states is a singleton set, and therefore the smallest possible.
We conclude that PT-symmetric observables alone cannot
yield a non-trivial theory
that extends standard quantum mechanics

We then studied the consequences of imposing quasi-Hermiticity 
on the set of observables. If all observables are quasi-Hermitian and not necessarily PT-symmetric, 
we found the resulting system to be equivalent to a standard quantum system.
 While this equivalence is known in the literature, our
approach using general probabilistic theories recovers this result from first principles
with no assumptions on the state space. 
We also investigated the GPT
in which observables are both PT-symmetric and quasi-Hermitian. We found these systems to 
be equivalent to real quantum theory systems. As real quantum theory is a restriction of standard
quantum theory \cite{Hickey,Scandolo-real1,Scandolo-real2}, this approach too fails to provide an extension of standard quantum mechanics.
Moreover, real quantum theory also faces the additional complication that the generator of time evolution  is not an observable of the theory, as noted in Ref.~\cite{Barnum-interference}.

In conclusion, neither PT-symmetry nor quasi-Hermiticity of observables 
leads to an extension of standard quantum mechanics. What possible constraints, if any,
could lead to such a meaningful extension remains an intriguing open question. 

\begin{acknowledgments}
SK is grateful for
 an Alberta Innovates Graduate
Student Scholarship. AA acknowledges support by Killam Trusts (Postdoctoral Fellowship). CMS acknowledges the support of the Natural Sciences and Engineering Research Council of Canada (NSERC) through the Discovery Grant ``The power of quantum resources'' RGPIN-2022-03025 and the Discovery Launch Supplement DGECR-2022-00119.
\end{acknowledgments}

\bibliographystyle{apsrev4-2}
\bibliography{ref.bib}


\end{document}